%% file: Red_Wu.tex
\documentclass[letterpaper, 10 pt, conference]{ieeeconf}

\usepackage[utf8]{inputenc}
\usepackage{graphicx}
\usepackage{amssymb}
\usepackage[cmex10]{amsmath}
\usepackage{color}
\usepackage{theorem}
\usepackage{array}
\usepackage{hyperref}
\usepackage{cite}
\usepackage{url}
\usepackage{nicefrac}
\usepackage{tikz}
\usepackage{pgfplots}
\usepackage{bm} 

\usepackage{xspace}
\usepackage{algorithm}
\usepackage{algpseudocode}
\usepackage{balance}

\usepackage{microtype}

\input{./Math_Defs}

\newcommand{\argmax}{\operatornamewithlimits{\arg\max}}

\algrenewcommand\algorithmicrequire{\textbf{Input:}}
\algrenewcommand\algorithmicensure{\textbf{Output:}}
\newcommand{\Fail}{\ensuremath{\mathsf{Fail}}\xspace}
\newcommand\T{^\star}
\renewcommand{\vec}[1]{\bm{#1}}

\newcommand{\errs}{\epsilon}
\newcommand{\errsL}{\errs_L}

\newcommand{\pr}{\,\mathalpha{:}\,} 

\IEEEoverridecommandlockouts

\begin{document}
\onecolumn 
\sloppy

\title{\LARGE \bf Reduced List-Decoding of Reed--Solomon Codes Using Reliability Information}


\author{Mostafa H. Mohamed$^{*}$, Johan S. R. Nielsen$^{*}$ and Martin Bossert$^{*}$
\thanks{This work has been supported by DFG, Germany under grant Bo 867/28.}
\thanks{* The authors are with Institute of Communications Engineering, Ulm University, Ulm 89081, Germany
        {\tt\small \{mostafa.h.mohamed, johan.nielsen, martin.bossert\}@uni-ulm.de}}%
}



\maketitle

\begin{abstract}
We decode Reed-Solomon codes using soft information provided at the receiver.
The Extended Euclidean Algorithm (EEA) is considered as an initial step to obtain an intermediate result.
The final decoding result is obtained by interpolating the output of the EEA at the least reliable positions of the received word.
We refer to this decoding method as reduced list-decoding, since not all received positions are used in the interpolation as in other list-decoding methods, such as the Guruswami-Sudan and Wu algorithms.
Consequently the complexity of the interpolation step is reduced considerably.
The probability of failure can be minimised by adjusting certain parameters, making it comparable with the K\"otter-Vardy algorithm but having a much lower complexity.
\end{abstract}

\section{Introduction}

Polynomial-time list-decoding of algebraic codes has been a lively research field for the last 15 years, since the Sudan algorithm \cite{Sudan}, and the generalisation of the Guruswami--Sudan algorithm (GSA) \cite{GS} demonstrated that such algorithms exist for Reed--Solomon codes.
The GSA corrects up to the Johnson bound $n - \sqrt{n(n-d)}$, which is a great improvement for low-rate codes; unfortunately a quite modest one for high-rate codes.

Also for Reed--Solomon codes, Wu \cite{Wu2008} presented an analytic extension of the classical Berlekamp--Massey algorithm (BMA) where decoding is done by ``rational interpolation'' using information from the output of the BMA, in case half the minimum distance-decoding fails.
The decoding radius turns out to be exactly the Johnson bound again.
The Wu algorithm first finds the error positions and can then by different means obtain the codeword or information, while the GSA finds the information word directly.
The Wu algorithm has worse complexity than the GSA for low-rate codes but better for medium to high-rate, though the asymptotic expressions are roughly the same \cite{NielsenPhD}.
It has since been shown how the initial step of the Wu decoder can be put into a more algebraic framework and e.g.~be carried out by the Extended Euclidean Algorithm (EEA), see \cite{Trif2010,WuNiel2013}.

The above algorithms are both hard-decision decoders.
The K\"otter--Vardy (KV) algorithm \cite{KV2003} is a generalisation of the GSA which incorporates soft information.
This celebrated algorithm has good performance but at the cost of relatively high complexity.

We present here a novel soft-decision list-decoder which generalises the Wu list-decoder (albeit in a different way than the KV from the GSA).
Our algorithm achieves a decoding performance close to that of the KV but retains a much lower complexity.
The core idea is to use the soft information to restrict the positions on which to perform the expensive rational interpolation.
We don't need to ``catch'' all errors in this restricted position set, but only some fraction of them, rendering the method fairly resilient to the quality of the reliability information.

Though our method returns a list of possible codewords, it is not similar to hard-decision list-decoding where the list contains all codewords in a certain radius.
Similar to the KV, the codewords in our list depend on the reliability information extracted from the received word.

We start with a few definitions and notation in Section \ref{S:Defs}.
After that, we revisit the hard-decision Wu decoder in Section \ref{S:Wu}, since it helps describe and understand our method.
In Section \ref{S:Rel}, we study the reliability information used for decoding.
For our simulations we have used a specific channel model and a particular method for calculating the reliability which is explained in detail.
We then give the novel reduced list-decoder in Section \ref{S:Red}.
Finally, in Section \ref{S:END}, we discuss simulation results and conclude.

\section{Definitions and Notations} \label{S:Defs}

Let $q$ be a power of a prime, $\Fq$ the finite field of order $q$. For $a, b \in \mathbb R$, we use the notation $[a;b] = \{ x \in \mathbb R \mid a \leq x \leq b \}$.

\begin{definition}[Generalised Reed--Solomon (GRS) code]\label{def:rs}
A $\mathcal{GRS}(n,k,d)$ code over a finite field $\Fq$ with $d=n-k+1$ is the set
{\small
\[ 
    \big\{ \big(  v_0C(\alpha_0),\ldots, v_{n-1}C(\alpha_{n-1}) \big)
        \mid C(x) \in \Fq[x] \land \deg C(x) < k \big\}
\]
}
for some $n$ distinct $\alpha_0, \ldots, \alpha_{n-1} \in \Fq$ as well as $n$ non-zero $v_0, \ldots, v_{n-1} \in \Fq$.
\end{definition}
The $\alpha_i$ are called \emph{evaluation points} and the $v_i$ \emph{column multipliers}. 

Let $\vec c = (c_0,c_1,\dots,c_{n-1}) \in \mathcal{GRS}$ be a codeword sent by the transmitter.
The receiver obtains $\vec r = \vec c + \vec e \in \Fq^n$ as a hard-decision received word, where $\vec e$ is an error.
The transmission model is explained in detail in Section \ref{S:Channel}, where we describe the nature of the reliability information accompanying $\vec r$.
The error positions, i.e.~the indexes where $\vec e$ is non-zero, are denoted by $E=\{i \mid c_i\ne r_i\}$ where $r_i$ are the elements of $\vec r$.
The number of errors in $\vec r$ are denoted $\errs = |E|$.
Our decoder is based on the Key Equation, where we, instead of seeking $\vec e$ or $\vec c$ directly, find two polynomials which encode $\vec e$:
\vspace{1em}
\begin{definition}
  The error locator $\Lambda(x)$ and the error evaluator $\Omega(x)$ are
  \begin{IEEEeqnarray*}{rCl+rCl}
    \Lambda(x) &=& \prod_{i \in E} (x-\alpha_i) ,\\
    \Omega(x) &=& \sum_{i \in E} e_i \hat v_i \prod_{j \in E \setminus \{i\}} (x-\alpha_j).
  \end{IEEEeqnarray*}
where $\hat v_i = (v_i\prod_{h \neq i}(\alpha_i-\alpha_h))^{-1}$.
\end{definition}
Note that $\Omega(x)$ and $\Lambda(x)$ are coprime, and that $\errs = \deg \Lambda(x) > \deg \Omega(x)$.

At the receiver, the syndrome polynomial $S(x)$ is calculated based on $\vec r$:
\begin{equation}\label{eqn:rsSyn}
  S(x) = \sum_{i=0}^{n-k-1} x^i\sum_{j=0}^{n-1} r_j\hat v_j\alpha_j^{d-2-i}.
\end{equation}

One can show that these polynomials are related by the Key Equation (see e.g. \cite[p. 362]{McWSl}):
\begin{equation}\label{eqn:RSkey}
    \Lambda(x)S(x) \equiv \Omega(x) \mod x^{d-1}.
\end{equation}
When $\errs < d/2$, then the above equation can be used to determine $\Lambda(x)$ and $\Omega(x)$ directly, e.g.~by using the BMA \cite[Ch. 7]{BerlekampBook} or the EEA \cite[p. 362]{McWSl}, \cite{BossBezz2013}.
Here we focus on the latter.
At the core of Wu's list-decoder is the observation that even when $\errs \geq d/2$, the result of running either the BMA or EEA still reveals crucial information about $\Lambda(x)$.
Wu proved these properties for the BMA \cite{Wu2008}; a similar approach can be used for the EEA, or one can use the theory of Gr\"o{}bner bases of $\Fq[x]$-modules.
For our purposes, the properties can be summed up in the following result, paraphrased from \cite[Proposition 5 and Section IV.B]{WuNiel2013}:

\begin{proposition}
  \label{prop:Lambda_partial}
  If the EEA is run on $x^{d-1}$ and $S(x)$ and halted at a certain iteration, we can extract from the intermediate polynomials at this point coprime polynomials $H_1(x)$ and $H_2(x)$ which satisfy\footnote{%
    More explicitly, recall that the Euclidean algorithm calculates quotients $Q_i(x)$ and remainders $R_i(x)$ ($i$ being the iteration index) \cite{BossBezz2013}.
    The EEA represents any remainder in the form of $R_i(x)=U_i(x)S(x) \mod x^{d-1}$, where the $U_i(x)$ can be calculated recursively using the quotients $Q_i(x)$.
    Then $H_1(x) = U_i(x)$ and $H_2(x) = U_{i-1}(x)$ for a particularly chosen iteration $i$.
  }:
  \begin{equation}
    \label{MainEqn}
    \Lambda(x) = A(x)H_1(x) + B(x)H_2(x),
  \end{equation}
  where $A(x)$ and $B(x)$ are unknown polynomials satisfying
  \begin{IEEEeqnarray*}{rCl}
    \deg A(x) &=& \deg \Lambda(x) - \deg H_1(x) ,\\
    \deg B(x) &\leq& \deg \Lambda(x) - d + \deg H_1(x).
  \end{IEEEeqnarray*}
\end{proposition}
In case of $\errs < d/2$, we see $\deg A(x) + \deg B(x) < 0$ so at least one of them is zero.
By some further properties of $H_1(x)$ and $H_2(x)$ one can show that it must be $B(x) = 0$, and that $\deg A(x) = 0$, which means $\Lambda(x)$ equals $H_1(x)$ up to a constant factor.
This is exactly classical hard-decision decoding up to half the minimum distance.

Wu generalised the noisy polynomial interpolation method by Guruswami and Sudan \cite{GS} into one for rational expressions.
We describe in the next section how this applies to finding $A(x)$ and $B(x)$ when $\errs \geq d/2$, but first we describe the general interpolation problem and solution.
Loosely, we can phrase the goal as follows: given points $(x_i, \beta_i)$ for $i = 1, \ldots, N$, one seeks two polynomials $f_1(x)$, $f_2(x) \in \Fq(x)$ and $\frac {f_1(x_i)}{f_2(x_i)} = \beta_i$ holds for many, but not necessarily all, values of $i$, and such that $\deg f_1(x)$ and $\deg f_2(x)$ are both small.
To properly handle that $f_2(x)$ might have roots among the $x_i$, Trifonov \cite{Trif2010} suggested to consider the $\beta_i$ as partially projective points $(y_i \pr z_i) \in \mathbb P^1_{\Fq}$, i.e.~that $(y_i \pr z_i) = (\lambda y_i \pr \lambda z_i)$ for all $\lambda \in \Fq^\star$ and where $(0 \pr 0)$ is disallowed.

The ingenious way to solve the problem is to construct a polynomial $Q \in \Fq[x][y, z]$, homogeneous in $y$ and $z$, in such a way that it is guaranteed that $Q(f_1, f_2)=0$.
The $f_1(x), f_2(x)$ can then be extracted from $Q$ as roots, which is possible since $y$ and $z$ are homogeneous in $Q$.
More precisely, the following theorem is a paraphrasing of \cite[Lemma 3]{Trif2010}:
\vspace{-1em}
\begin{theorem}[Rational Interpolation] \label{RatIntTheorem} \hspace{8em}
  Let $\ell$, $s$ and $T$ be positive integers, and let $\{(x_0,y_0 \pr z_0),(x_1,y_1 \pr z_1),\dots,(x_{N-1},y_{N-1} \pr z_{N-1})\}$ be $N \geq T$ points in $\Fq \times \mathbb P^1_{\Fq}$.
  Assume $Q(x,y,z)= \sum_{i=0}^{\ell} Q(x) y^i z^{\ell-i}$ is non-zero and such that $(x_i, y_i  \pr z_i)$ are zeroes of multiplicity $s$ for all $i=0,\dots,N-1$, and $\deg_{(1,w_1,w_2)}Q(x,y,z) < sT$, for two $w_1,w_2 \in \RR_+ \cup \{0\}$.
  Any two coprime polynomials $f_1(x),f_2(x)$ satisfying $\deg f_1(x) \le w_1, \deg f_2(x) \le w_2$, as well as, $z_i f_1(x_i)+y_i f_2(x_i)=0$ for at least $T$ values of $i$, will satisfy $Q(x,f_1(x),f_2(x)) = 0$.
\end{theorem}
In the above, $\deg_{(w_x,w_y,w_z)}$ is the $(w_x,w_y,w_z)$-weighted degree, i.e.~$\deg_{(w_x,w_y,w_z)} x^iy^jz^h = w_x i + w_y j + w_z h$ and for polynomials, it is the maximal weighted degree of its monomials.

The two integers $\ell$ and $s$, often referred to as the \emph{list size} and \emph{multiplicity} respectively, are not part of the rational interpolation problem one wishes to solve, but should simply be chosen in such a way to make it possible to construct the $Q$-polynomial.
One can regard the root-requirements to $Q$ as a linear system of equations in the monomials of $Q$, and the weighted degree constraints as a bound on the number of monomials available.
This gives a bound on the parameters on when it is guaranteed that a satisfactory $Q$ exists.
By analysis one can then conclude that this is the case whenever
\begin{equation} \label{Bound}
  T^2 > N(w_1 + w_2).
\end{equation}
Such an analysis along with precise choices of $s$ and $\ell$ can be found in \cite{Trif2010} or in more detail in \cite[Proposition 5.7]{NielsenPhD}.
In the latter, it is also proved that $s$ and $\ell$ can be chosen such that a satisfactory $Q$ exist under the additional requirement $\frac \ell s \geq \frac T w$; a subtle fact which we need.

\section{Review of Wu list-decoding} \label{S:Wu}

The hard-decision Wu list-decoder is now simply combining Proposition \ref{prop:Lambda_partial} with the tool of rational interpolation, Theorem \ref{RatIntTheorem}.

For this, the basic observation is that $\Lambda(x)$ evaluates to zero at all the error locations.
Thus, by Proposition \ref{prop:Lambda_partial}, we get for $i \in E$ that
\[
  A(\alpha_i)H_1(\alpha_i) + B(\alpha_i)H_2(\alpha_i) = 0.
\]
In other words, if we define $(x_i, y_i \pr z_i) = \big(\alpha_i, H_1(\alpha_i) \pr H_2(\alpha_i) \big)$ for $i = 1,\ldots, n$, then finding $A(x)$ and $B(x)$ is exactly a rational interpolation problem with $N = n$ and $T = \errs$.
Assuming that we knew the number of errors $\errs$, (\ref{Bound}) then tells us that we can find a $Q$ such that $Q(A(x), B(x)) = 0$  -- and therefore solve this problem -- as long as:
\begin{IEEEeqnarray*}{rCl+c+rCl}
  \errs^2 &>& n(2\errs - d)   & \iff &
  \errs < n - \sqrt{n(n-d)}.
\end{IEEEeqnarray*}
This is the Johnson bound, which is also the decoding radius of the GSA.

A caveat is of course that we do not know the number of errors.
The solution turns out to be surprisingly simple: we choose some decoding radius $\tau < n - \sqrt{n(n-d)}$, and we then construct a $Q$-polynomial for the worst possible case, i.e.~we set $T = \tau$.
It turns out that it can be shown -- see e.g.~\cite[Lemma 7]{WuNiel2013} -- that this $Q$-polynomial is also a valid interpolation polynomial for the rational interpolation problem where $T = \errs$ for any $\errs \leq \tau$, provided that $\ell/s \geq 1$.
Note that we are guaranteed to be able to find such $\ell$ and $s$ by our earlier remark since $T/w = \tau/(2\tau - d) > 1$ and $\tau < d$.

In Section \ref{S:Red}, we see that in the reduced list-decoding setting, the above concerns become more involved and lead to surprising behaviour.

The Wu list-decoder therefore consists of several non-trivial computational steps: the syndrome computation, the EEA, construction of $Q$, and root-finding in this $Q$.
In \cite{WuNiel2013}, it is shown how the complexity of all the steps can be completed in time $\mathcal{O}(\ell^M s n \log^{\mathcal{O}(1)}(\ell n))$, where $M \leq 3$ is the exponent for matrix multiplication.
This is the same as the fastest realisations of the GSA, see e.g.~\cite{cohn10} or \cite[Section 3.2]{NielsenPhD}.
One should be aware that the value of the parameter $s$ differs in the two methods (but the $\ell$ does not), see \cite[Section 5.2.2]{NielsenPhD}.

The aim of this paper is to let the Wu list-decoder take advantage of certain reliability information at the receiver, with the ultimate goal of reducing the complexity of the decoding while keeping decoding performance high.
We now introduce the form of this reliability information, as well as the channel model considered in the simulations.

\section{Reliability Information and Channel Model} \label{S:Rel}

Our decoding algorithm works as long as the following type of received information can be obtained: a hard-decision vector $\vec r \in \Fq^n$ as well as a reliability vector $\vec \eta = (\eta_0, \ldots, \eta_{n-1}) \in [0; 1]^n$, where $\eta_i$ is a measure of the probability that $r_i$ is the sent codeword symbol $c_i$.

How to obtain these two quantities, and how well the decoder will then finally perform, depends on the exact channel model.
For this paper, we use binary modulation over an Additive White Gaussian Noise (AWGN) channel for simulations. In the following, we describe the relevant details for this model.

\subsection{Channel Model} \label{S:Channel}

We now constrain ourselves to binary extension fields, i.e.~$q = 2^m$.
The sender wishes to transmit the codeword $\vec c \in \Ftwom^n$.
Using some given basis, he represents each $\Ftwom$ symbol as a vector in $\Ftwo^m$, and in turn uses BPSK modulation with $0$ mapped to $+1$ and $1$ to $-1$.
The resulting $nm$ symbols over $\{-1, +1\}$ which make up $\vec c$ are then individually transmitted over an AWGN channel.

The received signal has a Signal-to-Noise SNR $= E_s/(2\cdot N_0\cdot R)$, where $E_s$ is the energy of a single symbol, $N_0$ is the single-sided noise energy, and $R=\nicefrac{k}{n}$ is the code rate.
The raw output from the channel to the receiver is then a matrix $\vec y = [ y_{i,j} ] \in \RR^{n \times m}$, resulting from the sent $\pm 1$-symbols being perturbed by the noise.

\subsection{Reliability Calculation}
From $\vec y$, the receiver can calculate a matrix $\vec \rho = (\rho_{i, \beta}) \in [0; 1]^{n \times 2^m}$, where the columns of $\vec \rho$ are indexed by the elements of $\Ftwom$, and where $\rho_{i, \beta}$ is a measure of the probability that $c_i = \beta$, given the received matrix  $\vec y$.
For BPSK the relation between $\vec \rho$ and $\vec y$ is given by (see \cite[Section 7.4]{BossBook} or \cite{KampfWachterBoss2010}):

\begin{equation}\label{AW1}
\rho_{i,\beta}=\ln \frac{P(y_i|c_i=\beta)}{\displaystyle\sum_{l \ne j} P(y_i|c_i=l)}.
\end{equation}

Let $z_\beta=(z_{\beta,0},z_{\beta,k},\dots,z_{\beta,m-1}) \in \{-1, +1\}^{m}$ be the direct BPSK modulation of $\beta \in \Ftwom$ under the chosen basis over $\Ftwo$.
Since our channel is AWGN and memory-less, the above becomes:
\begin{equation}\label{AW3}
\rho_{i,\beta}=\ln \frac{\exp \left\lbrace \displaystyle\sum_{k=0}^{m-1} -\frac{1}{2\sigma^2}(y_{i,k}-z_{\beta,k})^2 \right\rbrace}{\displaystyle\sum_{l \ne \beta} \exp \left\lbrace \displaystyle\sum_{k=0}^{m-1} -\frac{1}{2\sigma^2}(y_{i,k}-z_{l,k})^2 \right\rbrace}.
\end{equation}

From $\vec \rho$, the receiver proceeds to extract both $\vec r$ and $\vec \eta$.
Firstly, $\vec r$ is chosen as $r_i = \argmax_\beta \{ \rho_{i, \beta} \}$, i.e.~the most likely symbol for each position, ties broken arbitrarily.

For $\vec{\eta}$, we use the principle from \cite{KanNishInazHira1994}:
\def\first{^{(1{\rm st})}}
\def\second{^{(2{\rm nd})}}
\[
 \eta_i = \rho_i\first - \rho_i\second,
\] 
where $\rho_i\first = \max_i\{ \rho_{i, \beta} \}$ and $\rho_i\second = \max_i \{ \rho_{i, \beta} \mid \rho_{i, \beta} \neq \rho_i\first \}$.

\section{Reduced List-Decoding} \label{S:Red}

What we want to do is very close to the algorithm of Section \ref{S:Wu}: we initially use the hard-decision guess $\vec r$ to compute a syndrome and run the EEA.
This succeeds in finding a codeword if $\vec r$ is less than $d/2$ errors away from one.
If this fails, we want to do rational interpolation; however, instead of using all $n$ points we restrict ourselves to only the least reliable positions, guided by $\vec \eta$.
How many positions are considered versus how many errors are ``caught'' inside these positions determine whether the rational interpolation succeeds.

More formally, introduce again $E, \errs$ and $\Lambda(x)$ based on $\vec r$, and assume that $\errs > d/2$.
According to Proposition \ref{prop:Lambda_partial}, by running the EEA on $S(x)$ and $x^{d-1}$ we can obtain polynomials $H_1(x), H_2(x)$ such that there exist some polynomials $A(x), B(x)$ with
\begin{equation}
  \label{eqn:LambdaFromHsRed}
  \Lambda(x) = A(x)H_1(x) + B(x)H_2(x)
\end{equation}
and $\deg A(x) = \errs - \deg H_1(x)$ and $\deg B(x) = \errs - \deg H_2(x)$.

Let $L$ be the number of positions to use for rational interpolation; it is not straightforward how to best choose this, but we get back to it later.
Without loss of generality, assume that reliabilities obtained from the previous section are sorted in an ascending order such that $\eta_0\le\eta_1\le\dots\le\eta_{n-1}$, such that the chosen positions are the $L$ first.

There are now two counters to keep in mind: $\errs$, the total number of errors, and $\errsL$, the number of errors present in the $L$ chosen positions.
The success of the rational interpolation depends on both of these. Likewise, choose two ``goal'' values: $\tau$ and $\tau_L \leq \tau$. They are \emph{loosely} the targeted value of $\errs$, respectively the number of errors there must be in the $L$ chosen positions for us to succeed. A visualisation of the parameters is shown in Figure \ref{fig:visual}.

\begin{figure}[t] 
\centering 
\includegraphics[scale=1.1]{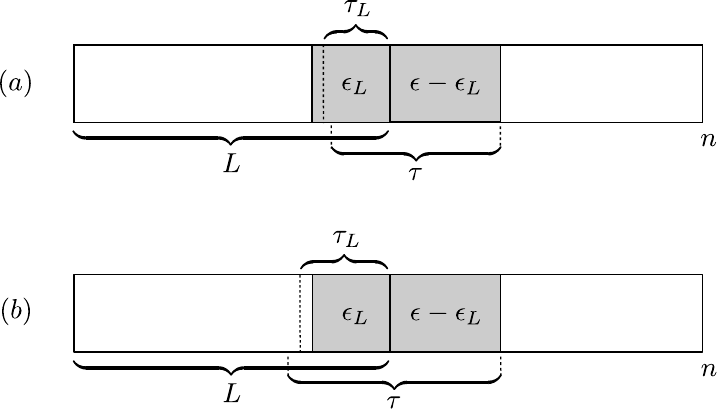}
\caption{Visualisation of $L$, $\epsilon$, $\epsilon_L$, $\tau$ and $\tau_L$. Error Positions have been ordered for overview. In $(a)$ $\epsilon>\tau$ and in $(b)$ $\epsilon<\tau$, but in either case decoding succeeds since $\epsilon_L$ is sufficiently large, as according to Lemma \ref{lem:errsL}.}\label{fig:visual}
\end{figure}

By Theorem \ref{RatIntTheorem}, and following arguments similar to those in Section \ref{S:Wu}, we then set up a rational interpolation problem for finding $A(x)$ and $B(x)$ using $H_1(x)$ and $H_2(x)$:
\begin{proposition}
  \label{prop:ReducedWu}
  Consider Theorem \ref{RatIntTheorem} and let $N = L$ and $(x_i, y_i \pr z_i) = \big(\alpha_i, H_1(\alpha_i) \pr H_2(\alpha_i)\big)$ for $i=0,\ldots, L-1$, as well as $T = \tau_L$, $w_1 = \tau - \deg H_1(x)$ and $w_2 = \tau - d + \deg H_2(x)$.
  If
  \begin{equation}
    \label{eqn:tauL}
    \tau_L^2 > L(2\tau-d) ,
  \end{equation}
  there exist valid choices of $s$ and $\ell$ such that a $Q(y,z)$ satisfying the requirements of Theorem \ref{RatIntTheorem} exists.
  Furthermore, if $\errs = \tau$ and $\errsL = \tau_L$, then $Q(A(x), B(x)) = 0$.
\end{proposition}
\begin{proof}
  The existence of $Q$ follows directly from Theorem \ref{RatIntTheorem} and (\ref{Bound}).
  The property $Q(x, A(x), B(x)) = 0$ follows from the arguments of Section \ref{S:Wu} since the equation $A(\alpha_i)H_1(\alpha_i) + B(\alpha_i)H_2(\alpha_i) = 0$ holds for $T = \tau_L$ out of $N=L$ values of $i$, and since $\deg A(x) + \deg B(x) \leq 2\errs - (\deg H_1(x) + \deg H_2(x)) = 2\tau-d = w_1+w_2$.
\end{proof}

\begin{algorithm}
  \caption{Reduced list-decoding with reliability}
  \label{alg:RedWu}
  \begin{algorithmic}[1]
    \Require A GRS code $\mathcal{C}$ over $\Fq$ with parameters $n,k,$ $d=n-k+1$ and evaluation points $ \alpha_0,\ldots,\alpha_{n-1}$.
    \hspace{5cm}
    A hard-decision guess $\vec r \in \Fq^n$ and reliability vector $\boldsymbol{\eta}=(\eta_0,\eta_1,\dots,\eta_{n-1})$.
    A decoding radius $\tau$ and the desired number of points to interpolate $L$.
    \Ensure A list of codewords in $\mathcal{C}$ or \Fail.
    \Statex 
    \State Calculate the syndrome $S(x)$ from $\vec r$ as in \eqref{eqn:rsSyn}.
    \State Run the EEA on $x^{d-1}, S(x)$ and calculate $H_1(x), H_2(x)$.
    \State If $H_1(x)$ is a valid error-locator of degree less than $d/2$, use it to correct $\vec r$, and if this yields a word in $\mathcal{C}$, return this one word.
    \State Otherwise, we seek $A(x), B(x)$ as in \eqref{MainEqn}.
    Set $\tau_L=\lfloor \sqrt{L(2\tau - d)} + 1\rfloor$ and set $w_1, w_2$ as in Proposition \ref{prop:ReducedWu}.
    Construct the $Q(y,z)$ described in that proposition, using satisfactory values of $s$ and $\ell$.
    \State Find all pairs of polynomials $A\T(x)$ and $B\T(x)$ such that $Q(A\T(x),B\T(x)) = 0$.
    Return \Fail if no such pairs exist.
    \State For each such pair, construct $\Lambda\T(x) = A\T(x) H_1(x) + B\T(x) H_2(x)$.
    If it is a valid error-locator, use it for correcting $\vec r$.
    Return \Fail if none of the factors yield error-locators.
    \State Return those of the corrected words that are in $\mathcal{C}$.
    Return \Fail if there are no such words.
  \end{algorithmic}
\end{algorithm}

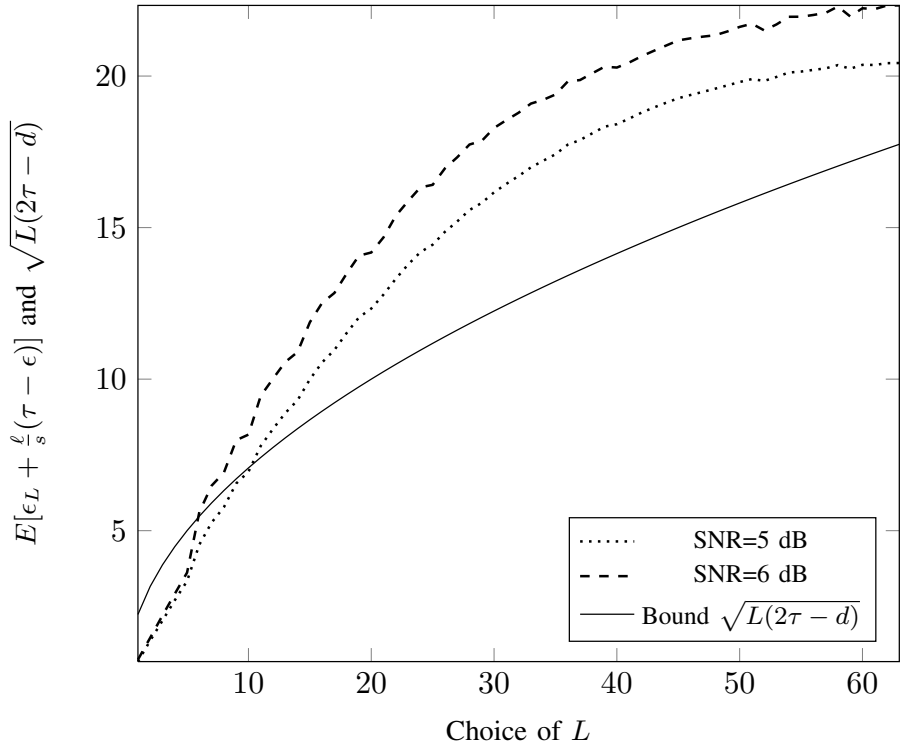
\begin{figure*}[t]
\begin{center}
\begin{tikzpicture}[scale=1.2]
\begin{axis}[enlargelimits=false,scale only axis,
xlabel=Choice of $L$,
ylabel=$E \mathbb{[} \epsilon_L+\tfrac \ell s(\tau - \errs) \mathbb{]} $ and $\sqrt{L(2\tau-d)}$,
legend pos= south east]

\pgfplotsset{
label style={font=\small},
legend style={font=\footnotesize}
}

\addplot[color=black, dotted, thick]
 table[x=i,y=SNR1] {./MTNS_2nd.data};
\addlegendentry{SNR=5 dB}


%
%

\addplot[color=black,dashed,thick]
 table[x=i,y=SNR5] {./MTNS_2nd.data};
\addlegendentry{SNR=6 dB}

\addplot[color=black]
 table[x=i,y=Bound] {./MTNS_2nd.data};
\addlegendentry{Bound $\sqrt{L(2\tau-d)}$}

\end{axis}
\end{tikzpicture}
\end{center}
\caption{$E[\errsL+\tfrac \ell s(\tau - \errs)]$ for $L$ least reliable positions vs $\sqrt{L(2\tau -d)}$.} \label{AVGtauL}
\end{figure*}
As before, we also have to consider if we succeed whenever $\errs \neq \tau$.
Here the situation is surprisingly different than in Section \ref{S:Wu}:
\begin{lemma}\label{lem:errsL}
  Considering Proposition \ref{prop:ReducedWu} when $\errs \neq \tau$, then\\ $Q(A(x), B(x)) = 0$ if:
  \begin{IEEEeqnarray*}{rCl+l:rCl}
    \frac{\ell}{s} & \geq & \frac{\tau_L-\errsL}{\tau-\errs}, & \textrm{whenever } & \tau & > & \errs,\\ 
    \frac{\ell}{s} & \leq & \frac{\errsL-\tau_L}{\errs-\tau}, & \textrm{whenever } & \tau & < & \errs.
  \end{IEEEeqnarray*}
\end{lemma}
\begin{proof}
  We prove the assertion by showing that $Q$ is a valid interpolation polynomial satisfying the requirements of Theorem \ref{RatIntTheorem} for almost the same rational interpolation problem but with $\hat T=\errsL$ and $\hat w_1 = \errs - \deg H_1(x)$ and $\hat w_2 = \errs + d - \deg H_2(x)$; the ``hats'' are added to distinguish these new parameters from those of Proposition \ref{prop:ReducedWu}. 
  For in that case $Q(x, A(x), B(x)) = 0$ follows from the theorem due to Proposition \ref{prop:Lambda_partial}.
 
  Since only $T, w_1$ and $w_2$ are changed in the newly considered rational interpolation, we only need to show that the new weighted-degree constraints on $Q$ are satisfied, i.e.~that:
  \begin{IEEEeqnarray*}{rCl+c}      
      \deg_{(1,\hat w_1,\hat w_2)}Q & < & s\errsL.
  \end{IEEEeqnarray*}
  Since $Q$ satisfied the original interpolation problem, then $\deg_{(1,w_1,w_2)} Q < s\tau_L$.
  We then compute
  \begin{IEEEeqnarray*}{rCl+c}
    \deg_{(1,\hat w_1,\hat w_2)}Q & \leq & \deg_{(1,w_1,w_2)}Q                                                                                       \\
                                  &      & \hspace{2em} -\min_{i=0,\ldots,\ell} \{\mbox{\scriptsize $i(w_1- \hat w_1) + (\ell - i)(w_2-\hat w_2)$}\} \\
                                  & = & \deg_{(1,w_1,w_2)}Q - \ell(\tau - \errs).
  \end{IEEEeqnarray*}
  Thus $\deg_{(1,\hat w_1,\hat w_2)}Q$ is satisfactory low whenever
  \[
    s\errsL \geq s\tau_L - \ell (\tau- \errs),
  \]
  which is equivalent to the conditions of the lemma.
\end{proof}

The above lemma thus reveals that when $\errs < \tau$, we can succeed, and we can even do so when fewer errors than $\tau_L$ are ``caught'' in the $L$ chosen positions.
But, perhaps surprisingly, it also reveals that when $\errs > \tau$, we can \emph{still} succeed as long as we also catch more errors in the chosen positions.

Algorithm \ref{alg:RedWu} is the complete proposed decoding algorithm.
The following theorem precisely characterises which codewords are returned:
\begin{theorem}
Let $\vec r$,$\vec \eta$ be the received word and its reliability vector respectively. Let $L,\ell,s,\tau,\tau_L$ be parameters as defined in Algorithm \ref{alg:RedWu}. If $\exists$ $\vec c \in \mathcal{C}$ s.t. $wt(\vec c-\vec r) < d/2$, then $\vec c$ is returned. 
Otherwise, the set T $\subset \mathcal{C}$ is returned s.t $\vec c \in T$ iff:

\[ wt(\vec c-\vec r)_L \ge \tau_L- \tfrac{\ell}{s} (\tau - wt(\vec c-\vec r) ), \]

where $wt(\vec x)_L$ is the number of non-zero entries of $\vec x$ within the least reliable $L$ positions according to $\vec \eta$.
\end{theorem}
\begin{proof}
  Follows from Proposition \ref{prop:ReducedWu} and Lemma \ref{lem:errsL}.
\end{proof}

Let us estimate the complexity:
up until Step 3, we have classical minimum distance decoding, which can be performed in $\mathcal O(n \log^{2+o(1)} n)$.
The remaining steps are as in regular Wu list-decoding, but where we only use $L$ of the total $n$ points.
The complexity of this part is therefore $\mathcal{O}(\ell^M s L \log^{\mathcal{O}(1)}(\ell L))$, with the same hidden constant as in Wu list-decoding.
As we see in the following section, $L$ can be chosen as only a fraction of $n$ while still ensuring good decoding performance.

\begin{figure*}[!t]
\begin{center}
\begin{tikzpicture}[scale=1.21]
\begin{semilogyaxis}[enlargelimits=false,scale only axis,
xlabel=SNR (dB),
ylabel=Block Failure Probability,
legend pos= south west]

\pgfplotsset{
label style={font=\small},
legend style={font=\footnotesize}
}

\addplot[color=black,thick]
 table[x=SNR,y=Classical] {./MTNS_1st.data};
\addlegendentry{Classical $\lfloor\nicefrac{d-1}{2}\rfloor=16$}

\addplot[color=black,loosely dotted,thick]
 table[x=SNR,y=Wu] {./MTNS_1st.data};
\addlegendentry{Wu, $\tau=19$}

\addplot[color=black,dotted, thick]
 table[x=SNR,y=KV] {./MTNS_1st.data};
\addlegendentry{K\"otter--Vardy}

\addplot[color=black,densely dashed, thin]
 table[x=SNR,y=Red_Wu_15] {./MTNS_1st.data};
\addlegendentry{Reduced Wu, $L=15$}

\addplot[color=black,thin]
 table[x=SNR,y=Red_Wu_25] {./MTNS_1st.data};
\addlegendentry{Reduced Wu, $L=25$}

\addplot[color=black,dashed,thick]
 table[x=SNR,y=Red_Wu_45] {./MTNS_1st.data};
\addlegendentry{Reduced Wu, $L=45$}


\end{semilogyaxis}
\end{tikzpicture}
\end{center}
\caption{Probability of failure vs. SNR.} \label{Pf}
\end{figure*}
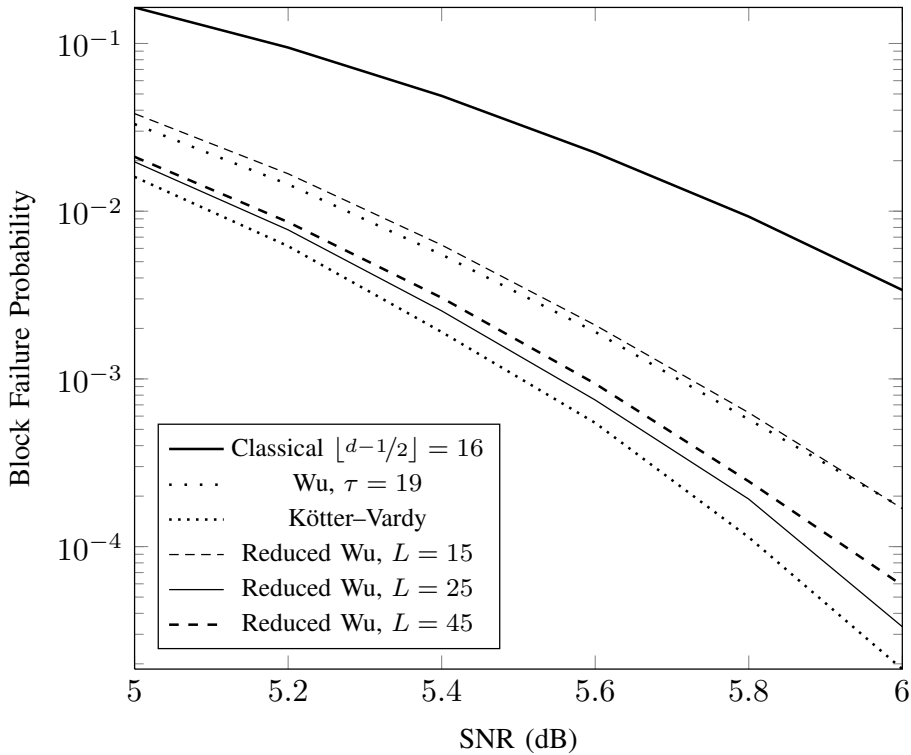

\subsection*{On the Choice of Parameters}

For the algorithm one needs to choose $L, \tau, \tau_L, \ell$ and $s$. We have not yet determined which choices lead to optimal performance.
However, our simulations indicate that the somewhat ad hoc choices we have made so far still gives good performance, and simultaneously show that a gain could definitely be made.
We here briefly go through some of our current observations.
Further analysis on the best possible choices and its impact on performance would be interesting.

Principally, in our simulations we have fixed $\tau$ and $L$ for every received word, but we expect that one could benefit from varying one or both parameters depending on the distribution of the reliabilities in $\vec\eta$.

The best choice of $\tau_L$ is easy, however.
We would like to succeed while having to catch as few errors as possible, and when $\errs = \tau$ this is clearly obtained by choosing $\tau_L$ minimally while satisfying \eqref{eqn:tauL}.
When $\errs \neq \tau$, we see from Lemma \ref{lem:errsL} that in all cases it is still beneficial to choose $\tau_L$ low.
Therefore, we can safely always choose $\tau_L$ minimally, i.e.~$\tau_L = \lfloor 
\sqrt{L(2\tau-d)} + 1 \rfloor$.

Given the parameters $L$ and $\tau_L$, one has many choices for $\ell$ and $s$ in the rational interpolation. One can use the numerically smallest possible values, as by the analysis of e.g.~Trifonov~\cite{Trif2010}; this is what we have done in our simulations since this minimises the computational complexity.
Choosing differently with the aim of changing the ratio $\ell/s$ might improve decoding performance at the price of increasing the complexity: according to Lemma \ref{lem:errsL}, when $\tau > \errs$, we would like $\ell/s$ to be as large as possible to maximise probability of success, while when $\tau < \errs$, we would like it as small as possible. Whichever has the biggest impact on the decoding performance depends on the code, $\tau$ and the SNR.

\section{Simulation Results and Conclusion} \label{S:END}

\subsection{Simulation}

We have performed simulations with the proposed decoding method and compared it to two other decoding methods: the hard-decision Wu list-decoder and the KV algorithm.
We have used an $\mathcal{RS}(63,31,33)$ code over $\F_{2^6}$ and an AWGN channel with reliability information as described in Section \ref{S:Rel}. 
Approximately one million codewords were simulated for each SNR.  
We chose $\tau = 19$ since this is the maximum possible hard-decision list-decoding radius of the Wu decoder for this code
(this is 3 errors beyond half the minimum-distance). The choice of the parameters $L$ and $\tau_L$ in the reduced list-decoding algorithm determines the decoding performance as shown later. 

From Lemma \ref{lem:errsL} it follows that decoding succeeds whenever:
\begin{equation*}
 \errsL + \tfrac \ell s (\tau - \errs) \geq \tau_L
\end{equation*}
We have therefore examined by simulation what the expected value of the above left-hand-side is; this can be seen in Figure \ref{AVGtauL}:
the solid curve is $\sqrt{L(2\tau-d)} \approx \tau_L$ as a function of $L$, while the two dotted lines correspond to the average of $\errsL + \tfrac \ell s (\tau - \errs)$ at two different SNRs.
When varying $L$ the choices of $\ell$ and $s$ make the ratio $\ell/s$ jump up and down; this is due to integer rounding on the small, possible values of $\ell$ and $s$.
This is the reason the dotted curves are not smooth.
As explained in the previous section, this possibly has an effect on the decoding performance.

When the simulated curve is below the target $\sqrt{L(2\tau-d)}$, then poor decoding performance can be expected.
Therefore, for this code and SNRs, at least $L > 10$ has to be chosen.
Intuitively, we expect the best performance when the simulated curve is as far as possible above the target.

We chose to simulate with three fixed choices of $L$: 15, 25 and 45.
In Figure \ref{Pf} the probability of failure (both wrong decoding and decoding failure) are plotted for the chosen code when varying the SNR from 5 to 6 dB.
The curves are compared with the classical hard-decision minimum-distance decoding, the hard-decision Wu list-decoding, and the KV soft-decision algorithm \cite{KV2003}.

The KV algorithm is the best performing soft-decision decoder for RS codes currently known.
As reliability it uses the $\vec \rho$ matrix described in Section \ref{S:Rel}.
It has a parameter, the multiplicity-sum, for adjusting the performance at the price of complexity.
In our simulations we have used a very high multiplicity-sum, namely $2n$.

Figure \ref{Pf} demonstrates that choosing a small $L=15$ with $\tau_L=9$, the performance of the reduced list-decoder is close to that of the hard-decision Wu list-decoder.
When we increase $L$ to $25$ with $\tau_L=12$, the performance takes a huge jump, and it comes close to the performance of the KV algorithm.
Further increasing $L=45$ with $\tau_L=12$ again decreases the performance.

The values of $\ell/s$ used in the rational interpolation $\nicefrac{5}{3}$ for $L=15$, $2$ for $L=25$ and $3$ for $L=45$.
We have not yet further investigated the performance impact of this difference.

\subsection{Conclusion}

Initially, our goal was to achieve the decoding performance of a hard-decision decoder with a soft decoder which could exploit reliability information to achieve lower complexity.
Surprisingly, the resulting decoder seems to be able to exceed this performance -- for some parameters, by far -- and get close to the KV algorithm.
Our method seems to excel at medium-rate codes where the benefits of hard-decision list-decoding is usually modest.

This is still with a complexity which is a lot lower than both the KV algorithm and the Wu list-decoder.

An open question for future investigation is to determine the best reachable performance of this algorithm, and how the parameters should be chosen to achieve this; in particular, it seems promising to choose $L$ and $\tau$ differently for each received word, depending on $\vec \eta$.

It would also be interesting to examine the algorithm's performance in other channel models and base fields.
Yet another possibility is to carry the reliability information of $\vec \eta$ into the rational interpolation, to achieve a method which more dynamically favours certain positions over others.

\balance


\bibliography{mhm}
\bibliographystyle{IEEEtran}


\end{document}

%% file: Math_Defs.tex
\newtheorem{theorem}{Theorem}
\newtheorem{definition}{Definition}
\newtheorem{lemma}{Lemma}

\newtheorem{proposition}{Proposition}

\newcommand{\Fq}{\mathbb F_{q}}
\newcommand{\Ftwo}{\mathbb F_{2}}
\newcommand{\Ftwom}{\mathbb F_{2^m}}
\newcommand{\F}{\mathbb F}

\newcommand{\RR}{\mathbb R}

\renewcommand{\hat}{\widehat}